\documentclass[a4paper,english,runningheads]{llncs}
\usepackage[T1]{fontenc}
\usepackage[latin9]{inputenc}
\usepackage{xcolor}
\usepackage{units}
\usepackage{amsmath}
\usepackage{amssymb}
\PassOptionsToPackage{normalem}{ulem}
\usepackage{ulem}

\makeatletter

\providecommand{\tabularnewline}{\\}
\providecolor{lyxadded}{rgb}{0,0,1}
\providecolor{lyxdeleted}{rgb}{1,0,0}

\usepackage[all]{xy}

\usepackage{cite}

\newcommand{\qPr}{\mathsf{qPr}}

\usepackage{babel}

\usepackage{babel}

\makeatother

\usepackage{babel}
\begin{document}

\title{Contextuality-by-Default 2.0: Systems with Binary Random Variables}

\titlerunning{Contextuality-by-Default 2.0}

\author{Ehtibar N. Dzhafarov\textsuperscript{1}, Janne V. Kujala\textsuperscript{2}}

\authorrunning{E. N. Dzhafarov, J. V. Kujala }

\institute{\textsuperscript{1}Purdue University\\
 ehtibar@purdue.edu\\
 $\,$\\
 \textsuperscript{2}University of Jyväskylä\\
 jvk@iki.fi}

\toctitle{Contextuality-by-Default Theory, Version 2.0}

\tocauthor{E. N. Dzhafarov, J. V. Kujala}
\maketitle
\begin{abstract}
The paper outlines a new development in the Contextuality-by-Default
theory as applied to finite systems of binary random variables. The
logic and principles of the original theory remain unchanged, but
the definition of contextuality of a system of random variables is
now based on multimaximal rather than maximal couplings of the variables
that measure the same property in different contexts: a system is
considered noncontextual if these multimaximal couplings are compatible
with the distributions of the random variables sharing contexts. A
multimaximal coupling is one that is a maximal coupling of any subset
(equivalently, of any pair) of the random variables being coupled.
Arguments are presented for why this modified theory is a superior
generalization of the traditional understanding of contextuality in
quantum mechanics. The modified theory coincides with the previous
version in the important case of cyclic systems, which include the
systems whose contextuality was most intensively studied in quantum
physics and behavioral sciences.

\end{abstract}

\section{Introduction: From maximality to multimaximality}

The Contextuality-by-Default (CbD) theory \cite{bookKD,conversations,DK2014Advances,DK2014LNCSQualified,DK2014Scripta,DKC_LNCS2016,DKL2015FooP,KDL2015PRL,deBarros,KDproof2016,DK_CCsystems}
was proposed as a generalization of the traditional contextuality
analysis in quantum physics \cite{Kurzynski2012,Kochen-Specker1967,Bell1964,Bell1966,15Fine,11Leggett,9CHSH,SuppesZanotti1981}.
The latter has been largely confined to \emph{consistently connected}
systems of random variables, those adhering to the \emph{``no-disturbance}''
principle \cite{Kurzynski2014,Ramanathan2012}: the distributions
of measurement outcomes remain unchanged under different measurement
conditions (\emph{contexts}). CbD allows for \emph{inconsistently
connected} systems, those in which context may influence the distribution
of measurement outcomes for one and the same property \cite{bookKD,conversations,DKC_LNCS2016,DKL2015FooP,KDL2015PRL,DKCZJ_2016,DZK_2015}.
In accordance with the CbD interpretation of the traditional contextuality
analysis, this generalization is achieved by replacing the \emph{identity
couplings} used in dealing with consistently connected systems by
\emph{maximal couplings}. 

Recall that, given a set of random variables $X,Y,\ldots,Z$, a coupling
of this set is any set of jointly distributed random variables, $\left(X',Y',\ldots,Z'\right)$,
with 
\[
X\sim X',\ Y\sim Y',\ \ldots,\ Z\sim Z',
\]
where $\sim$ stands for ``has the same distribution as.'' The coupling
$\left(X',Y',\ldots,Z'\right)$ is maximal if (using $\Pr$ as a symbol
for probability) the value of 
\[
p_{eq}=\Pr\left[X'=Y'=\ldots=Z'\right]
\]
is maximal possible among all possible couplings of $X,Y,\ldots,Z$.
The identity coupling is a special case of a maximal coupling, when
$p_{eq}=1$. The latter is possible if and only if all random variables
$X,Y,\ldots,Z$ (hence also $X',Y',\ldots,Z'$) are identically distributed:
\[
X\sim Y\sim\ldots\sim Z.
\]

The notion of a maximal coupling, however, is not the only possible
generalization of the identity couplings. And it has recently become
apparent that it is not the best possible generalization either. The
maximal-couplings-based definition of \sloppy{(non)contextual} systems
adopted in CbD does not have a certain intuitively plausible property
that is enjoyed by the identity-couplings-based definition of consistently
connected (non)contextual systems. This property is that \emph{any
subsystem of a consistently connected noncontextual system is noncontextual}.
A subsystem is obtained by dropping from a system some of the random
variables. An inconsistently connected noncontextual system in the previously published version of CbD (``CbD 1.0'') does not generally have this property: by dropping some of its
components one may be able to make it contextual.

In the new version, ``CbD 2.0,'' preservation of noncontextuality for subsystems is achieved
by replacing the notion of a maximal coupling in the definition of
(non)contextual systems by the notion of a \emph{multimaximal coupling}.
This term designates a coupling every subcoupling whereof is a maximal
coupling for the corresponding subset of the random variables being
coupled (see Definition \ref{def: multimax} below). 

The remainder of the paper is a systematic presentation of this idea
and of how it works in the analysis of contextuality. CbD 1.0 and
CbD 2.0 coincide when dealing with consistently connected systems
(as they must, because they both generalize this special case). They
also coincide when dealing with the important class of cyclic systems
\cite{DKL2015FooP,KDL2015PRL,KDproof2016} (see Section \ref{sec: Properties-of-contextuality}).
None of the principles upon which CbD is based changes in version
2.0 (Section \ref{sec: Contextuality-by-Default-theory:}). The recently
proposed logic of constructing a universal measure of contextuality
\cite{DK_CCsystems} also transfers to version 2.0 without changes
(Section \ref{sec: A-measure-of}).

\section{\label{sec: Contextuality-by-Default-theory:}Contextuality-by-Default
theory: Basics }

We briefly recapitulate here the main aspects of the Contextuality-by-Default
theory. We recommend, however, that the reader look through some of
the recent accounts of CbD 1.0, e.g., Refs. \cite{DKC_LNCS2016,conversations},
or (especially) Ref. \cite{DK_CCsystems}. 

Each random variable in CbD is double-indexed, $R_{q}^{c}$, where
$q$ is referred to as the \emph{content} of the random variables,
that which $R_{q}^{c}$ measures or responds to, and $c$ is referred
to as its \emph{context}, the conditions under which $R_{q}^{c}$
measures or responds to $q$. 
\begin{remark}
Following Ref. \cite{DK_CCsystems} we will write ``conteXt'' and
``conteNt'' to prevent their confusion in reading. The conteXt and
conteNt of a random variable uniquely identify it within a given system
of random variables. 
\end{remark}

Two random variables $R_{q}^{c}$ and $R_{q'}^{c'}$ are jointly distributed
if and only if they share a conteXt: $c=c'$. Otherwise they are \emph{stochastically
unrelated}. All random variables sharing a conteXt form a jointly
distributed $bunch$ of random variables. All random variables sharing
a conteNt form a \emph{connection}, the elements of which are pairwise
stochastically unrelated. It is necessary that all random variables
in a connection have the same set of possible values (more generally,
the same set and sigma-algebra). 

The present paper is primarily about systems in which all random variables
are binary. It is immaterial for contextuality analysis how these
values are named, insofar as they are identically named and identically
interpreted within each connection. For instance, if $R_{q}^{c}=1$
means ``spin-up along axis $z$ in particle 1'' and $R_{q}^{c}=2$
means ``spin-down along axis $z$ in particle 1,'' then all random
variables $R_{q}^{c'}$ ($c'\not=c$) should have the same possible
values, 1 and 2, with the same meanings. Note that for another conteNt
$q'$, the values of $R_{q'}^{c}$ need not be denoted in the same
way even if they have analogous interpretations: e.g., we may have
$R_{q'}^{c}=3=$ ``spin-up along axis $z$ in particle 2'' and $R_{q'}^{c}=4=$
``spin-down along axis $z$ in particle 2''.

The matrix below provides an example of a \emph{conteXt-conteNt system}
(c-c system) of random variables:

\begin{center}
\begin{tabular}{|c|c|c|c|c}
\cline{1-4} 
$R_{1}^{1}$$\begin{array}{cc}
\\
\\
\end{array}$ & $R_{2}^{1}$$\begin{array}{cc}
\\
\\
\end{array}$ & $\cdot$$\begin{array}{cc}
\\
\\
\end{array}$ & $R_{4}^{1}$$\begin{array}{cc}
\\
\\
\end{array}$ & $c_{1}$\tabularnewline
\cline{1-4} 
$R_{1}^{2}$$\begin{array}{cc}
\\
\\
\end{array}$ & $\cdot$$\begin{array}{cc}
\\
\\
\end{array}$ & $R_{3}^{2}$$\begin{array}{cc}
\\
\\
\end{array}$ & $\cdot$$\begin{array}{cc}
\\
\\
\end{array}$ & $c_{2}$\tabularnewline
\cline{1-4} 
$R_{1}^{3}$$\begin{array}{cc}
\\
\\
\end{array}$ & $R_{2}^{3}$$\begin{array}{cc}
\\
\\
\end{array}$ & $R_{3}^{3}$$\begin{array}{cc}
\\
\\
\end{array}$ & $R_{4}^{3}$$\begin{array}{cc}
\\
\\
\end{array}$ & $c_{3}$\tabularnewline
\cline{1-4} 
\multicolumn{1}{c}{$q_{1}$} & \multicolumn{1}{c}{$q_{2}$} & \multicolumn{1}{c}{$q_{3}$} & \multicolumn{1}{c}{$q_{4}$} & $\boxed{\boxed{\mathcal{R}_{ex}}}$\tabularnewline
\end{tabular}.
\par\end{center}

\noindent Each row here is a bunch of jointly distributed random
variables, each column is a connection (``between bunches''). Note
that not every conteNt should be measured in a given conteXt. 

The system $\mathcal{R}_{ex}$ can be conveniently used to illustrate
the logic of contextuality analysis. We first consider the connections
separately, and for each of them find all couplings that satisfy a
certain property $\mathsf{C}$. Let's call them \emph{$\mathsf{C}$-couplings}.
Then we determine if these $\mathsf{C}$-couplings are \emph{compatible}
with a coupling of the bunches of the c-c system (equivalently put,
with a coupling of the entire c-c system). 

The compatibility in question means the following. A coupling of (the
bunches of) the c-c system is a set of jointly distributed random
variables

\begin{center}
\begin{tabular}{|c|c|c|c|c}
\cline{1-4} 
$S_{1}^{1}$$\begin{array}{cc}
\\
\\
\end{array}$ & $S_{2}^{1}$$\begin{array}{cc}
\\
\\
\end{array}$ & $\begin{array}{cc}
\\
\\
\end{array}$ & $S_{4}^{1}$$\begin{array}{cc}
\\
\\
\end{array}$ & $c_{1}$\tabularnewline
\cline{1-4} 
$S_{1}^{2}$$\begin{array}{cc}
\\
\\
\end{array}$ & $\cdot$$\begin{array}{cc}
\\
\\
\end{array}$ & $S_{3}^{2}$$\begin{array}{cc}
\\
\\
\end{array}$ & $\cdot$$\begin{array}{cc}
\\
\\
\end{array}$ & $c_{2}$\tabularnewline
\cline{1-4} 
$S_{1}^{3}$$\begin{array}{cc}
\\
\\
\end{array}$ & $S_{2}^{3}$$\begin{array}{cc}
\\
\\
\end{array}$ & $S_{3}^{3}$$\begin{array}{cc}
\\
\\
\end{array}$ & $S_{4}^{3}$$\begin{array}{cc}
\\
\\
\end{array}$ & $c_{3}$\tabularnewline
\cline{1-4} 
\multicolumn{1}{c}{$q_{1}$} & \multicolumn{1}{c}{$q_{2}$} & \multicolumn{1}{c}{$q_{3}$} & \multicolumn{1}{c}{$q_{4}$} & $\boxed{\boxed{S_{ex}}}$\tabularnewline
\end{tabular},
\par\end{center}

\noindent such that
\[
\begin{array}{c}
\left(S_{1}^{1},S_{2}^{1},S_{4}^{1}\right)\sim\left(R_{1}^{1},R_{2}^{1},R_{4}^{1}\right),\\
\\
\left(S_{1}^{2},S_{3}^{2}\right)\sim\left(R_{1}^{2},R_{3}^{2}\right),\\
\\
\left(S_{1}^{3},S_{2}^{3},S_{3}^{3},S_{4}^{3}\right)\sim\left(R_{1}^{3},R_{2}^{3},R_{3}^{3},R_{4}^{3}\right).
\end{array}
\]
Since the elements of $S_{ex}$ are jointly distributed, the marginal
distributions of the columns corresponding to the connections of $\mathcal{R}_{ex}$
are well-defined:
\[
\begin{array}{ccc}
\left(S_{1}^{1},S_{1}^{2},,S_{1}^{3}\right) & \textnormal{is a coupling of connection} & R_{1}^{1},R_{1}^{2},R_{1}^{3},\\
\\
\left(S_{2}^{1},S_{2}^{3}\right) & \textnormal{is a coupling of connection} & R_{2}^{1},R_{2}^{3},\\
\\
\left(S_{3}^{2},S_{3}^{3}\right) & \textnormal{is a coupling of connection} & R_{3}^{2},R_{3}^{3},\\
\\
\left(S_{4}^{1},S_{4}^{3}\right) & \textnormal{is a coupling of connection} & R_{4}^{1},R_{4}^{3}.
\end{array}
\]
In CbD we pose the following question: is there a coupling $S_{ex}$
such that the subcouplings corresponding to the connections are\emph{
$\mathsf{C}$-}couplings? If the answer is affirmative, then we say
that the bunches of $\mathcal{R}_{ex}$ are compatible with at least
some of the combinations of the \emph{$\mathsf{C}$-}couplings for
its connections \textemdash{} and the c-c system is considered \emph{partially}
\emph{$\mathsf{C}$-noncontextual}. Otherwise, if no such a coupling
$S_{ex}$ exists, we say that the bunches of $\mathcal{R}_{ex}$ are
incompatible with any of the \emph{$\mathsf{C}$-}couplings for its
connections \textemdash{} and the c-c system is considered \emph{completely}
\emph{$\mathsf{C}$-contextual}. The intuition is that in a completely
\emph{$\mathsf{C}$-}contextual c-c system the conteXts ``interfere''
with one's ability to couple the measurements of each conteNt in a
specified (by $\mathsf{C}$) way \textemdash{} while the connections
can be coupled in this way if they are considered separately, ignoring
the conteXts.

The adjectives ``partially'' and ``completely'' do not belong
to the original theory. They are added here because one can also consider
a stronger (more restrictive) notion of noncontextual c-c systems
and, correspondingly, a weaker (less restrictive) notion of contextual
c-c systems. We say that a c-c system is \emph{completely $\mathsf{C}$-noncontextual
}if the bunches of $\mathcal{R}_{ex}$ are compatible with any combinations
of the\emph{ $\mathsf{C}$-}couplings for its connections; and it
is \emph{partially $\mathsf{C}$-contextual} if the bunches of $\mathcal{R}_{ex}$
are incompatible with at least some of these combinations. The intuition
is that in a completely \emph{$\mathsf{C}$-}noncontextual c-c system
the conteXts ``do not interfere'' in any way with \emph{$\mathsf{C}$-}couplings
of the measurements of any given conteNt (as if the connections were
taken separately, ignoring conteXts). 

In CbD 1.0 the $\mathsf{C}$-couplings are maximal couplings, as defined
in the opening paragraph of the paper. In CbD 2.0 $\mathsf{C}$-couplings
are multimaximal couplings, as defined below. We will see that if
all random variables in a system are binary and $\mathsf{C}$ is multimaximality,
then every connection has a unique $\mathsf{C}$-coupling (Theorem
\ref{thm: multimaximal-coupling}-Corollary \ref{cor: A-multimaximal-coupling}).
In this case the notions of partial and complete (non)contextuality
coincide, allowing us to drop these adjectives when speaking of (non)contextual
c-c systems.
\begin{remark}
It is important to accept that noncontextuality of a c-c system (even
if complete) does not mean that the conteXts are irrelevant and can
be ignored. On the contrary, they are relevant ``by defaults because,
e.g., $R_{2}^{1}$ and $R_{2}^{3}$ in the second connection of $\mathcal{R}_{ex}$
are distinct and stochastically unrelated random variables. Moreover,
the distributions of $R_{2}^{1}$ and $R_{2}^{3}$ may very well be
different (i.e., the c-c system may be inconsistently connected),
and this does not necessarily mean that the system is contextual (even
if only partially) in the sense of our definitions. The measurements
of the conteNt $q_{2}$ in conteXt $c_{3}$ can be ``directly''
influenced by the jointly-made measurements of $q_{3}$ (in which
case we can speak of ``signaling'' or ``disturbance''), while
in context $c_{1}$ this influence is absent \cite{bacciagaluppi,Kofler_2013}.
It is also possible that the experimental set-up in context $c_{3}$
is different from that in context $c_{1}$, in which case we can speak
of conteXt-dependent biases \cite{Nature_2011,Nature_2011_companion}.
All of this may account for the different distributions of $R_{2}^{1}$
and $R_{2}^{3}$, and none of this by itself makes the system contextual.
See Refs. \cite{conversations,DK_CCsystems,DKCZJ_2016} for argumentation
against confusing signaling and contextual biases with contextuality.
(Of course, if one so wishes, they can be called forms of contextuality,
but in a different sense from how contextulaity is understood in quantum
physics and in CbD.)
\end{remark}

\section{\label{sec: Multimaximal-couplings-for}Multimaximal couplings for
binary variables}
\begin{definition}
\label{def: multimax}\emph{Let $R_{q}^{1},\ldots,R_{q}^{k}$ ($k>1$)
be a connection of a system. A coupling $\left(T_{q}^{1},\ldots,T_{q}^{k}\right)$
of $R_{q}^{1},\ldots,R_{q}^{k}$ is a }multimaximal coupling\emph{
if, for any $m>1$ and any subset $\left(T_{q}^{i_{1}},\ldots,T_{q}^{i_{m}}\right)$
of $\left(T_{q}^{1},\ldots,T_{q}^{k}\right)$, the value of 
\[
\Pr\left[T_{q}^{i_{1}}=\ldots=T_{q}^{i_{m}}\right]
\]
is largest possible among all couplings of $R_{q}^{i_{1}},\ldots,R_{q}^{i_{m}}$.}
\end{definition}

The multimaximality plays the role of the constraint $\mathsf{C}$
in the definition of $\mathsf{C}$-couplings given in the previous
section. One finds multimaximal couplings for each of the connections
and then investigates their compatibility with the c-c system's bunches.

It is known that a maximal coupling exists for any connection \cite{Thor}.
This is not true for multimaximal couplings in general: such a coupling
need not exist if the number of possible values for the random variables
in a connection exceeds 2. 
\begin{example}
\label{exa: 1}Consider a connection consisting of random variables
$R_{q}^{1},R_{q}^{2},R_{q}^{3}$ each having values $1,2,3$ with
the following probabilities

\begin{center}%
\begin{tabular}{c|c|c|c|}
\multicolumn{1}{c}{} & \multicolumn{1}{c}{} & \multicolumn{1}{c}{} & \multicolumn{1}{c}{}\tabularnewline
\multicolumn{1}{c}{} & \multicolumn{1}{c}{$1$$ $} & \multicolumn{1}{c}{$2$$ $} & \multicolumn{1}{c}{$3$$ $}\tabularnewline
\cline{2-4} 
$R_{q}^{1}$$ $ & $0$$\begin{array}{c}
\\
\\
\end{array}$ & $\nicefrac{1}{2}$$\begin{array}{c}
\\
\\
\end{array}$ & $\nicefrac{1}{2}$$\begin{array}{c}
\\
\\
\end{array}$\tabularnewline
\cline{2-4} 
$R_{q}^{2}$$ $ & $\nicefrac{1}{2}$$\begin{array}{c}
\\
\\
\end{array}$ & $0$$\begin{array}{c}
\\
\\
\end{array}$ & $\nicefrac{1}{2}$$\begin{array}{c}
\\
\\
\end{array}$\tabularnewline
\cline{2-4} 
$R_{q}^{3}$$ $ & $\nicefrac{1}{2}$$\begin{array}{c}
\\
\\
\end{array}$ & $\nicefrac{1}{2}$$\begin{array}{c}
\\
\\
\end{array}$ & $0$$\begin{array}{c}
\\
\\
\end{array}$\tabularnewline
\cline{2-4} 
\multicolumn{1}{c}{} & \multicolumn{1}{c}{} & \multicolumn{1}{c}{} & \multicolumn{1}{c}{}\tabularnewline
\end{tabular}$\:$.\end{center}

\noindent If a multimaximal coupling $\left(T_{q}^{1},T_{q}^{2},T_{q}^{3}\right)$
exists, we should have (see Ref. \cite{Thor}, or Theorem 3.3 in Ref.
\cite{DK_CCsystems})

\[
\begin{array}{ccccc}
\Pr[T_{q}^{1}=T_{q}^{2}=1]=0 &  & \Pr[T_{q}^{1}=T_{q}^{2}=2]=0 &  & \Pr[T_{q}^{1}=T_{q}^{2}=3]=0.5\\
\\
\Pr[T_{q}^{2}=T_{q}^{3}=1]=0.5 &  & \Pr[T_{q}^{2}=T_{q}^{3}=2]=0 &  & \Pr[T_{q}^{2}=T_{q}^{3}=3]=0\\
\\
\Pr[T_{q}^{1}=T_{q}^{3}=1]=0 &  & \Pr[T_{q}^{1}=T_{q}^{3}=2]=0.5 &  & \Pr[T_{q}^{1}=T_{q}^{3}=3]=0
\end{array}
\]
from which we have in particular

\[
\Pr[T_{q}^{1}=T_{q}^{2}=3]=\Pr[T_{q}^{1}=T_{q}^{3}=2]=\Pr[T_{q}^{2}=T_{q}^{3}=1]=0.5.
\]
But these three events are pairwise mutually exclusive, so the sum
of their probabilities cannot exceed 1. \hfill$\square$
\end{example}

It can also be shown that, in the case of random variables with more
than two possible values, a multimaximal coupling, if it exists, is
not generally unique.
\begin{example}
\label{exa: 2}Consider a connection consisting of random variables
$R_{q}^{1},R_{q}^{2},R_{q}^{3}$ each having one of six values (denoted
$1,1',2,2',3,3'$) with the following probabilities

\begin{center}%
\begin{tabular}{c|c|c|c|c|c|c|}
\multicolumn{1}{c}{} & \multicolumn{1}{c}{} & \multicolumn{1}{c}{} & \multicolumn{1}{c}{} & \multicolumn{1}{c}{} & \multicolumn{1}{c}{} & \multicolumn{1}{c}{}\tabularnewline
\multicolumn{1}{c}{} & \multicolumn{1}{c}{$1$$ $} & \multicolumn{1}{c}{$1'$} & \multicolumn{1}{c}{$2$$ $} & \multicolumn{1}{c}{$2'$} & \multicolumn{1}{c}{$3$$ $} & \multicolumn{1}{c}{$3'$$ $}\tabularnewline
\cline{2-7} 
$R_{q}^{1}$$ $ & $0$$\begin{array}{c}
\\
\\
\end{array}$ & 0$\begin{array}{c}
\\
\\
\end{array}$ & 0$ $ & $\nicefrac{1}{2}$$ $ & 0$ $ & $\nicefrac{1}{2}$$ $\tabularnewline
\cline{2-7} 
$R_{q}^{2}$$ $ & $0$$\begin{array}{c}
\\
\\
\end{array}$ & $\nicefrac{1}{2}$$\begin{array}{c}
\\
\\
\end{array}$ & $0$$ $ & $0$$ $ & $\nicefrac{1}{2}$$ $ & $0$$ $\tabularnewline
\cline{2-7} 
$R_{q}^{3}$$ $ & $\nicefrac{1}{2}$$\begin{array}{c}
\\
\\
\end{array}$ & 0$\begin{array}{c}
\\
\\
\end{array}$ & $\nicefrac{1}{2}$$ $ & $0$$ $ & $0$$ $ & $0$$ $\tabularnewline
\cline{2-7} 
\multicolumn{1}{c}{} & \multicolumn{1}{c}{} & \multicolumn{1}{c}{} & \multicolumn{1}{c}{} & \multicolumn{1}{c}{} & \multicolumn{1}{c}{} & \multicolumn{1}{c}{}\tabularnewline
\end{tabular}$\:$.\end{center}

\noindent Then the distinct couplings whose distributions are shown
below,

\begin{center}%
\begin{tabular}{cccc}
 &  &  & \tabularnewline
$\left(\dot{T}_{q}^{1},\dot{T}_{q}^{2},\dot{T}_{q}^{3}\right)=$ & $\left(2',1',1\right)$$ $ & $\left(3',3,2\right)$ & otherwise\tabularnewline
\cline{2-4} 
\multicolumn{1}{c|}{prob. mass} & \multicolumn{1}{c|}{$\nicefrac{1}{2}$} & \multicolumn{1}{c|}{$\nicefrac{1}{2}$} & \multicolumn{1}{c|}{0}\tabularnewline
\cline{2-4} 
 &  &  & \tabularnewline
\end{tabular}

and

\begin{tabular}{cccc}
 &  &  & \tabularnewline
$\left(\ddot{T}_{q}^{1},\ddot{T}_{q}^{2},\ddot{T}_{q}^{3}\right)=$ & $\left(2',3,2\right)$$ $ & $\left(3',1',1\right)$ & otherwise\tabularnewline
\cline{2-4} 
\multicolumn{1}{c|}{prob. mass} & \multicolumn{1}{c|}{$\nicefrac{1}{2}$} & \multicolumn{1}{c|}{$\nicefrac{1}{2}$} & \multicolumn{1}{c|}{0}\tabularnewline
\cline{2-4} 
 &  &  & \tabularnewline
\end{tabular}$\:,$\end{center}

are both multumaximal couplings. \hfill$\square$
\end{example}

However, the situation is different if the random variables in a connection
are all binary: multimaximal couplings in this case always exist and
are unique. In the theorem to follow we denote the values of all variables
$R_{q}^{i}$ by $1,2$, and we will write values of $\left(T_{q}^{1},\ldots,T_{q}^{k}\right)$
as strings of 1's and 2's, without commas.
\begin{theorem}
\label{thm: multimaximal-coupling}Let $R_{q}^{1},\ldots,R_{q}^{k}$
be a connection with binary random variables arranged so that the
values of $p_{i}=\Pr\left[R_{q}^{i}=1\right]$ are sorted $p_{1}\leq\ldots\leq p_{k}$.
Then $\left(T_{q}^{1},\ldots,T_{q}^{k}\right)$ is a multimaximal
coupling of $R_{q}^{1},\ldots,R_{q}^{k}$ if and only if all values
of $\left(T_{q}^{1},\ldots,T_{q}^{k}\right)$ are assigned zero probability
mass, except for 
\[
\left[\begin{array}{ccc}
\textnormal{value of }\left(T_{q}^{1},\ldots,T_{q}^{k}\right) & | & \textnormal{probability mass}\\
11\ldots1 & | & p_{1}\\
21\ldots1 & | & p_{2}-p_{1}\\
22\ldots1 & | & p_{3}-p_{2}\\
\vdots & | & \vdots\\
\overset{l}{\overbrace{2\ldots2}}\underset{k-l}{\underbrace{1\ldots1}} & | & p_{l+1}-p_{l}\\
\vdots & | & \vdots\\
22\ldots2 & | & 1-p_{k}
\end{array}\right].\tag{\ensuremath{\dagger}}
\]
\end{theorem}

\begin{proof}
Note that the distribution of $\left(T_{q}^{1},\ldots,T_{q}^{k}\right)$
in the theorem's statement is well-defined, and that $\left(T_{q}^{1},\ldots,T_{q}^{k}\right)$
is indeed a coupling of $R_{q}^{1},\ldots,R_{q}^{k}$: for any $1\leq l\leq k$,
\[
\Pr\left[T_{q}^{l}=1\right]=\sum_{m=0}^{l-1}\Pr\left[\overset{m}{\overbrace{2\ldots2}}\underset{k-m}{\underbrace{1\ldots1}}\right]=\sum_{m=0}^{l-1}\left(p_{m+1}-p_{m}\right)=p_{l}=\Pr\left[R_{q}^{l}=1\right].
\]

\emph{Sufficiency}. The ``if'' part is checked directly: for any
$1\leq i_{1}<\ldots<i_{m}\leq k$,
\[
\begin{array}{r}
\Pr\left[T_{q}^{i_{1}}=\ldots=T^{i_{m}}=1\right]=\sum_{m=0}^{i_{1}-1}\Pr\left[\overset{m}{\overbrace{2\ldots2}}\underset{k-m}{\underbrace{1\ldots1}}\right]\\
\\
=\sum_{m=0}^{i_{1}-1}\left(p_{m+1}-p_{m}\right)=p_{i_{1}}=\Pr\left[T_{q}^{i_{1}}=1\right],
\end{array}
\]
which is the maximal possible value for the leftmost probability.
Analogously, 
\[
\begin{array}{c}
\Pr\left[T_{q}^{i_{1}}=\ldots=T^{i_{m}}=2\right]=\sum_{m=i_{m}}^{k}\Pr\left[\overset{m}{\overbrace{2\ldots2}}\underset{k-m}{\underbrace{1\ldots1}}\right]\\
\\
=\sum_{m=i_{m}}^{k}\left(p_{m+1}-p_{m}\right)=1-p_{i_{m}}=\Pr\left[T_{q}^{i_{m}}=2\right],
\end{array}
\]
which is also the maximal possible probability. This establishes that
$\left(T_{q}^{i_{1}},\ldots,T^{i_{m}}\right)$ is a maximal coupling
for $\left(R_{q}^{i_{1}},\ldots,R^{i_{m}}\right)$. 

\emph{Necessity.} The ``only if'' part of the statement is proved
by (i) observing that $\Pr\left[22\ldots2\right]=1-p_{k}$, and (ii)
proving that if $l$ is the ordinal position of the first 1 in the
value of $\left(T_{q}^{1},\ldots,T_{q}^{k}\right)$, then

\[
\Pr\left[\overset{l-1}{\overbrace{2\ldots2}}\underset{k-l+1}{\underbrace{1\ldots1}}\right]=p_{l}-p_{l-1},
\]
and for all other strings with the first 1 in the $l$th position
the probabilities are zero. We prove (ii) by induction on $l$. For
$l=1$, we have 
\[
p_{1}=\Pr\left[11\ldots1\right].
\]
Since
\[
p_{1}=\Pr\left[T_{q}^{1}=1\right]=\Pr\left[11\ldots1\right]+\sum\Pr\left[1\underset{\textnormal{not all 1's}}{\underbrace{\ldots}}\right],
\]
all the summands under the summation operator must be zero. Let the
statement be proved up to and including $l<k$. We have 
\[
p_{l+1}=\Pr\left[T_{q}^{l+1}=\ldots=T_{q}^{k}=1\right]=\Pr\left[\overset{l}{\overbrace{2\ldots2}}\underset{k-l}{\underbrace{1\ldots1}}\right]+\sum\Pr\left[\overset{\mbox{not all 2's}}{\overbrace{\ldots}}\underset{k-l}{\underbrace{1\ldots1}}\right].
\]
By the induction hypothesis, all summands under the summation operator
are zero, except for
\[
\left[\begin{array}{ccc}
\textnormal{value of }\left(T_{q}^{1},\ldots,T_{q}^{k}\right) & | & \textnormal{probability mass}\\
11\ldots1 & | & p_{1}\\
21\ldots1 & | & p_{2}-p_{1}\\
22\ldots1 & | & p_{3}-p_{2}\\
\vdots & | & \vdots\\
\overset{l-1}{\overbrace{2\ldots2}}\underset{k-l+1}{\underbrace{1\ldots1}} & | & p_{l}-p_{l-1}
\end{array}\right].
\]
These values sum to $p_{l}$. Hence
\[
\Pr\left[\overset{l}{\overbrace{2\ldots2}}\underset{k-l}{\underbrace{1\ldots1}}\right]=p_{l+1}-p_{l}.
\]
We also have
\[
\begin{array}{r}
p_{l+1}=\Pr\left[T_{q}^{l+1}=1\right]=\Pr\left[\overset{l}{\overbrace{2\ldots2}}\underset{k-l}{\underbrace{1\ldots1}}\right]+\sum\Pr\left[\overset{\mbox{not all 2's}}{\overbrace{\ldots}}\underset{k-l}{\underbrace{1\ldots1}}\right]\\
\\
+\sum\Pr\left[\overset{l+1}{\overbrace{\ldots1}}\underset{\textnormal{not all 1's}}{\underbrace{\ldots}}\right]\\
\\
=\left(p_{l+1}-p_{l}\right)+p_{l}+\sum\Pr\left[\overset{l+1}{\overbrace{\ldots1}}\underset{\textnormal{not all 1's}}{\underbrace{\ldots}}\right]
\end{array}
\]
whence the summands under the last summation operator must all be
zero. \hfill$\square$
\end{proof}

\begin{corollary}
\label{cor: A-multimaximal-coupling}A multimaximal coupling $\left(T_{q}^{1},\ldots,T_{q}^{k}\right)$
exists and is unique for any connection $R_{q}^{1},\ldots,R_{q}^{k}$
with binary random variables.
\end{corollary}

The significance of this result is that insofar as we confine our
analysis to c-c systems of binary random variables, every bunch (a
row in a c-c matrix) has a known distribution and every connection
(a column in the c-c matrix) has a uniquely imposed on it distribution.
The only question is whether the distributions along the rows and
along the columns of a c-c matrix are mutually compatible, i.e., can
be viewed as marginals of an overall coupling of the entire c-c system. 

We can now formulate the CbD 2.0 definition of (non)contextuality
in systems with binary random variables.
\begin{definition}
\label{def: multimaximal binary}\emph{A coupling of a c-c system
is called }multimaximally connected\emph{ if every subcoupling of
this coupling corresponding to a connection of the system is a multimaximal
coupling of this connection.}
\end{definition}

\begin{definition}
\label{def: (non)contextual binary}\emph{A c-c system of binary random
variables is }noncontextual\emph{ if it has a multimaximally connected
coupling. Otherwise it is }contextual\emph{.}
\end{definition}

\begin{remark}
As explained in the next section, any (non)contextual system of binary
random variables is completely (non)contextual. Because of this it
is unnecessary to use the qualification ``completely'' in the definition
above. Note that this definition applies only to systems of binary
random variables. The extension of this definition to arbitrary random
variables is not unique, and we leave this topic outside the scope
of this paper (but will discuss it briefly in Section \ref{sec: Conclusion:-How-to}).
\end{remark}

\section{\label{sec: Properties-of-contextuality}Properties of contextuality}

Contextuality analysis of the systems of binary random variables is
simplified by the following theorem, proved in Ref. \cite{DK_FdP}.
\begin{theorem}
\label{thm: pairs binary}Let $R_{q}^{1},\ldots,R_{q}^{k}$ be a connection
with binary random variables arranged so that the values of $p_{i}=\Pr\left[R_{q}^{i}=1\right]$
are sorted $p_{1}\leq\ldots\leq p_{k}$. Then $\left(T_{q}^{1},\ldots,T_{q}^{k}\right)$
is a multimaximal coupling of $R_{q}^{1},\ldots,R_{q}^{k}$ if and
only if $\left(T_{q}^{i},T_{q}^{i+1}\right)$ is a maximal coupling
of $\left\{ R_{q}^{i},R_{q}^{i+1}\right\} $ for $i=1,\ldots,k-1$.
\end{theorem}

In other words, in the case of binary random variables, multimaximality
can be defined in terms of certain pairs of random variables rather
than all possible subsets thereof, as it was done in Definition \ref{def: multimax}.
As shown in Section \ref{sec: Conclusion:-How-to} below, a pairwise
formulation can also be used in the general case, for arbitrary random
variables. 

The main motivation for switching from the maximal couplings of CbD
1.0 to multimaximal couplings is to be able to prove the following
theorem.
\begin{theorem}
\label{thm: subsystems binary}In a noncontextual c-c system of binary
random variables every subsystem (obtained from the system by removing
from it some of the random variables) is noncontextual.
\end{theorem}

\begin{proof}
Let $S$ be a multimaximally connected coupling of a system $\mathfrak{\mathcal{R}}$.
Let $\mathfrak{\mathcal{R}}'$ be a system obtained by deleting a
random variable $R_{q}^{c}$ from $\mathcal{R}$; and let $S'$ be
the set of random variables obtained by deleting from $S$ the corresponding
random variable $S_{q}^{c}$. Then $S'$ is a multimaximally connected
coupling of $\mathcal{R}'$. Indeed, $S'$ is jointly distributed,
its subcouplings corresponding to the system's bunches have the same
distributions as these bunches (including the bunch for conteXt $c$),
and its subcouplings corresponding to the system's connections are
multimaximal couplings (including the connection for context $q$,
by the definition of a multimaximal coupling). \hfill$\square$
\end{proof}

There are other desirable properties of the revised definition of
contextuality.

First of all we should mention a property shared by CbD 1.0 and 2.0,
one that should hold for any reasonable definition of contextuality.
If a c-c system is consistently connected (i.e., $R_{q}^{c}\sim R_{q}^{c'}$
for all $q,c,c'$ such that $q$ is measured in both $c$ and $c'$),
then the system is (non)contextual if and only if it is (non)contextual
in the traditional sense (as interpreted in CbD): the multimaximal
couplings for connections consisting of identically distributed random
variables are identity couplings. 

Another property worth mentioning is that, using the terminology introduced
at the end of Section \ref{sec: Contextuality-by-Default-theory:},
whether a c-c system of binary random variables is contextual or noncontextual,
it is always completely contextual (respectively, completely noncontextual).
This follows from the fact that multimaximal couplings for connections
consisting of binary random variables are unique, whence if the combination
of these unique couplings is (in)compatible with the system's bunches
then it is \emph{all} combinations of the couplings that are (in)compatible
with the system's bunches.

A third property we find important follows from the fact that if a
connection contains just two random variables, then their maximal
coupling is their multimaximal coupling. As a result, the theory of
contextuality for cyclic c-c systems \cite{DKL2015FooP,KDL2015PRL,Klyachko,DKCZJ_2016}
remains unchanged. Recall that a cyclic c-c system of binary random
variables is one in which (1) any bunch consists of two random variables,
and (2) any connection consists of two random variables (and, without
loss of generality, the c-c system cannot be decomposed into two disjoint
cyclic c-c systems). The conteXt-conteNt matrix below shows a cyclic
system with 3 conteNts and 3 conteXts (their numbers in a cyclic system
are always the same, and called the \emph{rank} of the c-c system):

\begin{center}
\begin{tabular}{|c|c|c|c}
\cline{1-3} 
$R_{1}^{1}$$\begin{array}{cc}
\\
\\
\end{array}$ & $R_{2}^{1}$$\begin{array}{cc}
\\
\\
\end{array}$ & $\cdot$$\begin{array}{cc}
\\
\\
\end{array}$ & $c_{1}$\tabularnewline
\cline{1-3} 
$\cdot$$\begin{array}{cc}
\\
\\
\end{array}$ & $R_{2}^{2}$$\begin{array}{cc}
\\
\\
\end{array}$ & $R_{3}^{2}$$\begin{array}{cc}
\\
\\
\end{array}$ & $c_{2}$\tabularnewline
\cline{1-3} 
$R_{5}^{1}$$\begin{array}{cc}
\\
\\
\end{array}$ & $\cdot$$\begin{array}{cc}
\\
\\
\end{array}$ & $R_{3}^{3}$$\begin{array}{cc}
\\
\\
\end{array}$ & $c_{3}$\tabularnewline
\cline{1-3} 
\multicolumn{1}{c}{$q_{1}$} & \multicolumn{1}{c}{$q_{2}$} & \multicolumn{1}{c}{$q_{3}$} & $\boxed{\boxed{\mathcal{CYC}_{3}}}$\tabularnewline
\end{tabular}$\:$.
\par\end{center}

\noindent A prominent example of a noncyclic c-c system each of whose
connections consist of two binary random variables is one derived
from the Cabello-Estebaranz-Alcaine proof \cite{Cabello_PhysicsLettr1996}
of the Kochen-Specker theorem in 4D space: the system there consists
of 36 random variables arranged into 9 bunches (shown below by columns)
containing 4 random variables each, and 18 connections (shown by rows)
containing two random variables each:
\[
\left[\begin{array}{cccccccccc}
 & c_{1} & c_{2} & c_{3} & c_{4} & c_{5} & c_{6} & c_{7} & c_{8} & c_{9}\\
q_{0001} & \star & \star\\
q_{0010} & \star &  &  &  & \star\\
q_{1100} & \star &  & \star\\
q_{1200} & \star &  &  &  &  &  & \star\\
q_{0100} &  & \star &  &  & \star\\
q_{1010} &  & \star &  &  &  &  &  & \star\\
q_{1020} &  & \star &  & \star\\
q_{1212} &  &  & \star & \star\\
q_{1221} &  &  & \star &  &  & \star\\
q_{0011} &  &  & \star &  &  &  & \star\\
q_{1111} &  &  &  & \star &  & \star\\
q_{0102} &  &  &  & \star &  &  &  & \star\\
q_{1001} &  &  &  &  & \star &  &  &  & \star\\
q_{1002} &  &  &  &  & \star & \star\\
q_{0120} &  &  &  &  &  & \star &  &  & \star\\
q_{1121} &  &  &  &  &  &  & \star & \star\\
q_{1112} &  &  &  &  &  &  & \star &  & \star\\
q_{2111} &  &  &  &  &  &  &  & \star & \star
\end{array}\right]\:.
\]
 Here, the star symbol in the cell defined by conteXt $c_{i}$ and
conteNt $q_{j}$ designates a binary random variable $R_{j}^{i}$
(the quadruple index at $q$ represents a ray in a 4D real Hilbert
space, as labeled in Ref. \cite{Cabello_PhysicsLettr1996}). The contextual
analysis of such systems generalizes the 4D version of the Kochen-Specker
theorem in the same way (although computationally more demanding)
in which cyclic c-c systems of rank 3,4,5 generalize the treatment
of, respectively, the Suppes-Zanotti-Leggett-Garg \cite{11Leggett,SuppesZanotti1981},
EPR-Bohm-Bell \cite{Bell1964,9CHSH,15Fine}, and Klyachko-Can-Binicoglu-Shumovsky
systems \cite{Klyachko}. More general proofs of the Kochen-Specker
theorem (e.g., by Peres \cite{Peres1995}) translate into systems
with more than two binary random variables per connection. The multimaximal-couplings-based
analysis here will yield different results from the maximal-couplings-based
one.

\section{\label{sec: A-measure-of}A measure of contextuality}

In accordance with the linear consistency theorem proved in Ref. \cite{DK_CCsystems},
a c-c system of random variables always has a \emph{quasi-coupling}
that agrees with a given set of couplings imposed on its connections.
Let us clarify this. 

A \emph{quasi-random variable} $X$ is defined by assigning to its
possible values real numbers (not necessarily nonnegative) that sum
to 1. These numbers are called \emph{quasi-probabilit}y \emph{masses},
or simply \emph{quasi-probabilities}. For instance, a variable $X$
with values 1 and 2 to which we assign quasi-probabilities $\qPr\left[X=1\right]=-5$,
$\qPr\left[X=2\right]=6$ is a quasi-random variable. A quasi-random
variable is a proper random variable if and only if the quasi-probabilities
assigned to its values are nonnegative. If a quasi-random variable
$X$ is a vector, $\left(X_{1},\ldots,X_{n}\right)$, it can be referred
to as a vector of jointly distributed quasi-random variables, even
if each $X_{i}$ is a proper random variable. A vector of jointly
distributed quasi-random variables may very well have marginals (subvectors)
that are proper random vectors. 

A quasi-coupling of a c-c system $\mathcal{R}$ is a vector $S$ of
jointly distributed quasi-random variables in a one-to-one correspondence
with the elements of $\mathcal{R}$, such that every subcoupling of
$S$ that corresponds to a bunch of the system has a (proper) distribution
that coincides with that of the bunch. Finally, the quasi-coupling
$S$ agrees with a set of multimaximal couplings of the system's connections
if any subcoupling of $S$ that corresponds to a connection has the
same (proper) distribution as this connection's multimaximal coupling.

As an example, consider again our c-c system $\mathcal{R}_{ex}$:

\begin{center}
\begin{tabular}{|c|c|c|c|c}
\cline{1-4} 
$R_{1}^{1}$$\begin{array}{cc}
\\
\\
\end{array}$ & $R_{2}^{1}$$\begin{array}{cc}
\\
\\
\end{array}$ & $\cdot$$\begin{array}{cc}
\\
\\
\end{array}$ & $R_{4}^{1}$$\begin{array}{cc}
\\
\\
\end{array}$ & $c_{1}$\tabularnewline
\cline{1-4} 
$R_{1}^{2}$$\begin{array}{cc}
\\
\\
\end{array}$ & $\cdot$$\begin{array}{cc}
\\
\\
\end{array}$ & $R_{3}^{2}$$\begin{array}{cc}
\\
\\
\end{array}$ & $\cdot$$\begin{array}{cc}
\\
\\
\end{array}$ & $c_{2}$\tabularnewline
\cline{1-4} 
$R_{1}^{3}$$\begin{array}{cc}
\\
\\
\end{array}$ & $R_{2}^{3}$$\begin{array}{cc}
\\
\\
\end{array}$ & $R_{3}^{3}$$\begin{array}{cc}
\\
\\
\end{array}$ & $R_{4}^{3}$$\begin{array}{cc}
\\
\\
\end{array}$ & $c_{3}$\tabularnewline
\cline{1-4} 
\multicolumn{1}{c}{$q_{1}$} & \multicolumn{1}{c}{$q_{2}$} & \multicolumn{1}{c}{$q_{3}$} & \multicolumn{1}{c}{$q_{4}$} & $\boxed{\boxed{\mathcal{R}_{ex}}}$\tabularnewline
\end{tabular}.
\par\end{center}

\noindent Let all random variables be binary. Then, as we know, each
connection has a unique multimaximal coupling. Let us denote these
couplings (going from the leftmost column to the rightmost one in
the matrix) 
\[
\left(T_{1}^{1},T_{1}^{2},T_{1}^{3}\right),\left(T_{2}^{1},T_{2}^{3}\right),\left(T_{3}^{2},T_{3}^{3}\right),\left(T_{4}^{1},T_{4}^{3}\right).
\]
The theorem mentioned in the opening line of this section says that
one can always find a quasi-coupling $S$ for $\mathcal{R}_{ex}$,

\begin{center}
\begin{tabular}{|c|c|c|c|c}
\cline{1-4} 
$S_{1}^{1}$$\begin{array}{cc}
\\
\\
\end{array}$ & $S_{2}^{1}$$\begin{array}{cc}
\\
\\
\end{array}$ & $\begin{array}{cc}
\\
\\
\end{array}$ & $S_{4}^{1}$$\begin{array}{cc}
\\
\\
\end{array}$ & $c_{1}$\tabularnewline
\cline{1-4} 
$S_{1}^{2}$$\begin{array}{cc}
\\
\\
\end{array}$ & $\cdot$$\begin{array}{cc}
\\
\\
\end{array}$ & $S_{3}^{2}$$\begin{array}{cc}
\\
\\
\end{array}$ & $\cdot$$\begin{array}{cc}
\\
\\
\end{array}$ & $c_{2}$\tabularnewline
\cline{1-4} 
$S_{1}^{3}$$\begin{array}{cc}
\\
\\
\end{array}$ & $S_{2}^{3}$$\begin{array}{cc}
\\
\\
\end{array}$ & $S_{3}^{3}$$\begin{array}{cc}
\\
\\
\end{array}$ & $S_{4}^{3}$$\begin{array}{cc}
\\
\\
\end{array}$ & $c_{3}$\tabularnewline
\cline{1-4} 
\multicolumn{1}{c}{$q_{1}$} & \multicolumn{1}{c}{$q_{2}$} & \multicolumn{1}{c}{$q_{3}$} & \multicolumn{1}{c}{$q_{4}$} & $\boxed{\boxed{S_{ex}}}$\tabularnewline
\end{tabular},
\par\end{center}

\noindent such that
\[
\begin{array}{c}
\left(S_{1}^{1},S_{1}^{2},S_{1}^{3}\right)\sim\left(T_{1}^{1},T_{1}^{2},T_{1}^{3}\right),\\
\\
\left(S_{2}^{1},S_{2}^{3}\right)\sim\left(T_{2}^{1},T_{2}^{3}\right),\\
\\
\left(S_{3}^{2},S_{3}^{3}\right)\sim\left(T_{3}^{2},T_{3}^{3}\right),\\
\\
\left(S_{4}^{1},S_{4}^{3}\right)\sim\left(T_{4}^{1},T_{4}^{3}\right).
\end{array}
\]
Clearly, the system $\mathcal{R}_{ex}$ is noncontextual if and only
if among all such quasi-couplings $S_{ex}$ there is at least one
proper coupling. 

It is convenient for our purposes to look at this in the following
way (introduced in Ref. \cite{DK_CCsystems} but derived from an idea
proposed in Ref. \cite{DeBarrosOas2014}). For each quasi-coupling
$S_{ex}$ one can compute its \emph{total variation}. The latter is
defined as the sum of the absolute values of all quasi-probabilities
assigned to the values of $S_{ex}$ (i.e., to all $2^{9}$ combinations
of values of $S_{1}^{1},S_{2}^{1},\ldots,S_{4}^{3}$). If $S_{ex}$
is a proper coupling, this total variation equals 1, otherwise it
is greater than 1. Therefore, if the system $\mathcal{R}_{ex}$ is
contextual, then the total variation of its quasi-couplings is always
greater than 1. As shown in Ref. \cite{DK_CCsystems}, one can always
find a quasi-coupling $S_{ex}^{*}$ of $\mathcal{R}_{ex}$ that has
the smallest possible value of the total variation. This value (perhaps,
less 1, if one wants zero rather than 1 to be the smallest value)
can be taken to be a \emph{measure of contextuality}. 

Generalizing, we have the following statement.
\begin{theorem}
Any c-c system of binary random variables has a quasi-coupling whose
subcouplings corresponding to the system's connections are their multimaximal
couplings. Among all such quasi-couplings there is at least one with
the smallest possible value of total variation (which value is then
considered a measure of contextuality for the system).
\end{theorem}

\section{\label{sec: Conclusion:-How-to}Conclusion: How to generalize}

For c-c systems with binary random variables multimaximal couplings
are definitely a better way of generalizing identity couplings of
the traditional contextuality analysis than maximal couplings. A system
that is deemed noncontextual in terms of multimaximal couplings has
noncontextual subsystems. The contextuality of a contextual system
and noncontextuality of a noncontextual system are both complete if
one uses multimaximal couplings to define them. And the theory specializes
to the previous version (CbD 1.0) when applied to cyclic systems and
to other systems whose connections consist of pairs of random variables.

The question to pose now is what one should do with non-binary random
variables. The most straightforward way to construct a general theory
is to simply drop the qualification ``binary'' in Definition \ref{def: (non)contextual binary}.
There are, however, some complications associated with this approach.
Connections involving non-binary variables may not have multimaximal
couplings (Section \ref{sec: Multimaximal-couplings-for}) One has
to decide whether such systems are contextual, and how to measure
the degree of contextuality in them if they are. Another complication,
shared with the CbD 1.0, is that multimaximal couplings are not unique
if the random variables are not all binary, because of which one no
longer can ignore the difference between complete and partial forms
of (non)contextuality. Conceptual and computational adjustments have
to be made. 

At the same time, some of the properties mentioned in Section \ref{sec: Properties-of-contextuality}
hold for arbitrary random variables, at least for categorical ones
(those with finite number of values). Theorem \ref{thm: subsystems binary}
obviously holds for arbitrary random variables if noncontextuality
is taken to be partial. The definition of the (non)contextuality of
a system of random variables reduces to the traditional one when a
system is consistently connected. Theorem \ref{thm: pairs binary}
also generalizes to arbitrary random variables, although in a somewhat
weaker form due to the loss of the linear ordering of the distributions
within a connection.
\begin{theorem}
Let $R_{q}^{1},\ldots,R_{q}^{k}$ be a connection. Then $\left(T_{q}^{1},\ldots,T_{q}^{k}\right)$
is a multimaximal coupling of $R_{q}^{1},\ldots,R_{q}^{k}$ if and
only if $\left(T_{q}^{c},T_{q}^{c'}\right)$ is a maximal coupling
of $\left\{ R_{q}^{c},R_{q}^{c'}\right\} $ for all $c<c'$ in $\left\{ 1,\ldots,k\right\} $.
\end{theorem}

\begin{proof}
The ``only if'' part is true because pairs are subsets. To prove
the ``if'' part, assume the contrary: there is a subset of the connection
(without loss of generality, the connection itself, $R_{q}^{1},\ldots,R_{q}^{k}$)
such that its coupling $\left(T_{q}^{1},\ldots,T_{q}^{k}\right)$
is not maximal while $\Pr\left[T_{q}^{c}=T_{q}^{c'}\right]$ is maximal
possible for all $c,c'$. Then, by the theorem on maximal couplings
(see Ref. \cite{Thor} or Ref. \cite{DK_CCsystems}, Theorem 3.3)
there is a value $v$ in the common set of values for all random variables
$T_{q}^{c}$ such that 
\[
\Pr\left[T_{q}^{1}=T_{q}^{2}=\ldots=T_{q}^{k}=v\right]<\min_{c\in\left\{ 1,\ldots,n\right\} }\left(\Pr\left[T_{q}^{c}=v\right]\right),
\]
while, for any $c,c'\in\left\{ 1,\ldots,k\right\} $, 
\[
\Pr\left[T_{q}^{c}=T_{q}^{c'}=v\right]=\min\left(\Pr\left[T_{q}^{c}=v\right],\Pr\left[T_{q}^{c'}=v\right]\right).
\]
Then, by replacing each $T_{q}^{c}$ with 
\[
\widetilde{T}_{q}^{c}=\left\{ \begin{array}{ccc}
1 & if & T_{q}^{c}=v\\
2 & if & otherwise
\end{array}\right.,
\]
and considering $\left(\widetilde{T}_{q}^{1},\ldots,\widetilde{T}_{q}^{k}\right)$
a coupling for some connection consisting of binary random variables,
we come to a contradiction with Theorem \ref{thm: pairs binary}.
\hfill$\square$
\end{proof}

There is a complication, however, that seems especially serious for
simply dropping the qualification ``binary'' in Definition \ref{def: (non)contextual binary}:
this approach allows a noncontextual system of random variables to
become contextual under \emph{coarse-graining}. The latter means lumping
together some of the values of the variables constituting some of
the connections. Thus, if $R_{q}^{c}$ has values $1,2,3,4$, one
could lump together 1 and 2 and obtain a random variables with three
values (and do the same for all other random variables in the connection
for conteNt $q$). It is natural to expect that a system should preserve
its noncontextuality under such course-graining, but this is not the
case generally.
\begin{example}
The system consisting of the single connection with six values ($1,1',2,2',3,3'$)
in Example \ref{exa: 2} is noncontextual, because it does have multimaximal
couplings. However, if one lumps together $i$ and $i'$ and denotes
the lumped value $i$ ($=1,2,3$), one obtains the system considered
in Example \ref{exa: 1}, which is contextual because it does not
have a multimaximal coupling. \hfill$\square$
\end{example}

A radical solution for all the problems mentioned is to deal with
binary random variables only. This can be achieved by replacing each
non-binary random variable $R_{q}^{c}$ in a system with a bunch of
jointly distributed dichotomizations thereof (that thereby becomes
a sub-bunch of the bunch representing conteXt $c$). For instance,
if $R_{q}^{c}$ has values $1,2,3,4$, then it could be represented
by $2^{4-1}-1=7$ jointly distributed binary random variables. The
joint distribution is very simple: of the $2^{7}$ values of this
bunch all but 4 have zero probability masses. Of course, every other
random variable with conteNt $q$ should be dichotomized in the same
way, replacing thereby the corresponding connection with 7 new connections.
Coarse-graining in this approach becomes a special case of extracting
from a system a subsystem. The price one pays for the conceptual simplicity
thus achieved is a great increase of the numbers of random variables
in each bunch (becoming infinite if the original system involves non-categorical
random variables), although the cardinality of the supports of the
bunches remains unchanged. It is to be seen if this dichotomization
approach proves feasible.

\subsubsection*{Acknowledgments.}

This research has been supported by NSF grant SES-1155956 and AFOSR
grant FA9550-14-1-0318.  We are grateful to Victor H. Cervantes for his critical comments on the manuscript.

\end{document}